%% file: main.tex
\documentclass[letterpaper, 10 pt, conference]{ieeeconf}  

\IEEEoverridecommandlockouts                              

\usepackage{graphicx} 
\usepackage{amsmath} 
\usepackage{mathtools}
\usepackage{amssymb}  
\usepackage{algorithm}
\usepackage{algpseudocode}
\usepackage{comment}
\usepackage{xcolor}
\usepackage{subcaption}
\usepackage{url}
\usepackage{enumerate}
\usepackage{cite}
\usepackage{flushend}
\usepackage{optidef}
\usepackage[font=small]{caption}

\algtext*{EndFor}
\algtext*{EndIf}

\newtheorem{theorem}{Theorem}
\newtheorem{aim}[theorem]{Aim}
\newtheorem{assumption}[theorem]{Assumption}
\newtheorem{corollary}{Corollary}[theorem]

\newtheorem{remark}[theorem]{Remark}

\newcommand{\yfHighlight}[1]{\textcolor{black}{#1}}

\title{\LARGE \bf
Collision-free Motion Planning for Mobile Robots by Zero-order Robust Optimization-based MPC}

\author{Yunfan Gao$^{1, 2}$\quad Florian Messerer$^{2}$\quad Jonathan Frey$^{2, 3}$\quad Niels van Duijkeren$^1$\quad Moritz Diehl$^{2,3}$
\thanks{$^{1}$ Robert Bosch GmbH, Corporate Research, Stuttgart, Germany
	{\tt\small \{yunfan.gao, niels.vanduijkeren\}@de.bosch.com}}
\thanks{$^{2}$ Department of Microsystems Engineering (IMTEK), University of Freiburg, Germany
	{\tt\small \{florian.messerer, jonathan.frey, moritz.diehl\}@imtek.uni-freiburg.de}}
\thanks{$^{3}$ Department of Mathematics, University of Freiburg, Germany}
\thanks{The research that led to this paper was funded by Robert Bosch GmbH. This work was also supported by 
           DFG via Research Unit FOR 2401 and project 424107692 on Robust MPC and by the EU via ELO-X 953348.
        }
}
\begin{document}

\maketitle
\thispagestyle{empty}
\pagestyle{empty}

\begin{abstract}
This paper presents an implementation of robust model predictive control (MPC) for collision-free reference trajectory tracking for mobile robots. 
The presented approach considers the robot motion to be subject to process noise bounded by ellipsoidal sets.
In order to efficiently handle the evolution of the disturbance ellipsoids within the MPC,
  the zero-order robust optimization (zoRO) scheme is applied~\cite{Zanelli2021a}.
The idea is to fix the disturbance ellipsoids within one optimization iteration
  and solve the problem repeatedly with updated disturbance ellipsoid trajectories.
The zero-order approach is suboptimal in general. 
  However, we show that it does not impair convergence to the reference trajectory 
    in the absence of obstacles. 
The experiments on an industrial mobile robot prototype demonstrate the performance of the controller.

\end{abstract}

\input{sections/1-introduction.tex}
\input{sections/2-problem.tex}
\input{sections/3-preliminary.tex}
\input{sections/4-zoMPC.tex}
\input{sections/5-experiments.tex}

\section{CONCLUSIONS}
\label{sec:conclusion}

This paper presents a tailored implementation of the zero-order robust optimization-based approach from~\cite{Zanelli2021a} 
  to collision-free reference trajectory tracking.
Given a mobile robot and a reference trajectory,
  we have devised a controller which achieves robust collision avoidance, 
  robust convergence to the reference trajectory in absence of obstacles, 
  and real-time feasible computation time on typical robot computing platforms.
The effectiveness of the module has been demonstrated on the real industrial mobile robot.
Future work might focus on strict real-time capabilities and the extension to non-static obstacles.


\bibliographystyle{./IEEEtran} 
\bibliography{refs}

\end{document}

%% file: sections/1-introduction.tex
\section{INTRODUCTION}

Model predictive control (MPC) has received increasing attention in the field planning and control for mobile robots~\cite{Berntorp2017,Roesmann2021,limon2021tracking}.
It predicts robot trajectories and optimizes robot motion based on an objective function with constraint satisfaction.
Although impressive results have been obtained using MPC, 
  the constraint satisfaction is not robust against disturbances.
In general, as optimal solutions are often at the edge of constraints, 
  small disturbances can lead to constraint violation.
It is especially problematic in robot motion planning since violations of anti-collision constraints may impair safety.

Robust MPC explicitly models and propagates disturbances~\cite{MAYNE2005219}.
Under the assumption that the disturbance is bounded, 
  robust MPC places tubes containing all possible disturbed states around the nominal trajectory.
The constraints are tightened such that all realizations within the tube fulfill the nominal constraints.

Solving robust MPC problems efficiently and non-conservatively is challenging.
Ellipsoidal sets are often used to approximate the sets of all possible disturbed states~\cite{Houska2011}.
The drawback is that the tube-based approach typically introduces a number of variables quadratic in the number of nominal system states, 
  thus significantly increasing the computational demand. 
The zero-order algorithm and the adjoint inexact algorithm were proposed~\cite{Zanelli2021a, feng2020inexact}
   to efficiently solve the robust MPC problem with ellipsoidal sets.
It reduces the computational complexity 
  by keeping the ellipsoidal sets fixed in an optimization iteration and 
    updating the ellipsoidal sets afterwards.
As some optimization variables are fixed, 
  the necessary optimality conditions change and 
    the solutions of the zero-order variant are suboptimal in general~\cite{Zanelli2021a}.

\begin{figure*}[t]
	\centering
	\vspace{2pt}
	\begin{minipage}[t]{0.82\textwidth}
		\centering
		\includegraphics[width=\textwidth]{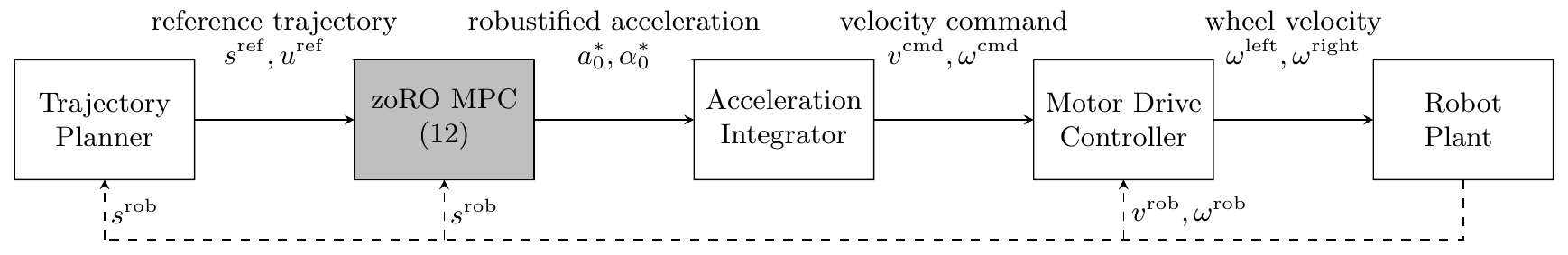}
		\caption{Diagram of the control pipeline. 
			The original reference trajectory $s^{\mathrm{ref}}, u^{\mathrm{ref}}$ is robustified by the zero-order robust optimization (zoRO) MPC. 
			The integrator integrates the robustified acceleration $a^*_0, \alpha^*_0$
			and outputs the velocity commands $v^{\mathrm{cmd}}, \omega^{\mathrm{cmd}}$. 
			The differential drive controller computes the angular velocities of the robot wheels to be tracked by the motor controllers.
			The robot state $s^{\mathrm{rob}}$ is fed back to close the loop.
		}
		\label{fig:control_pipeline}
	\end{minipage}
	\hspace{0.01\textwidth}
	\begin{minipage}[t]{0.15\textwidth}
		\centering
		\includegraphics[width=0.85\columnwidth]{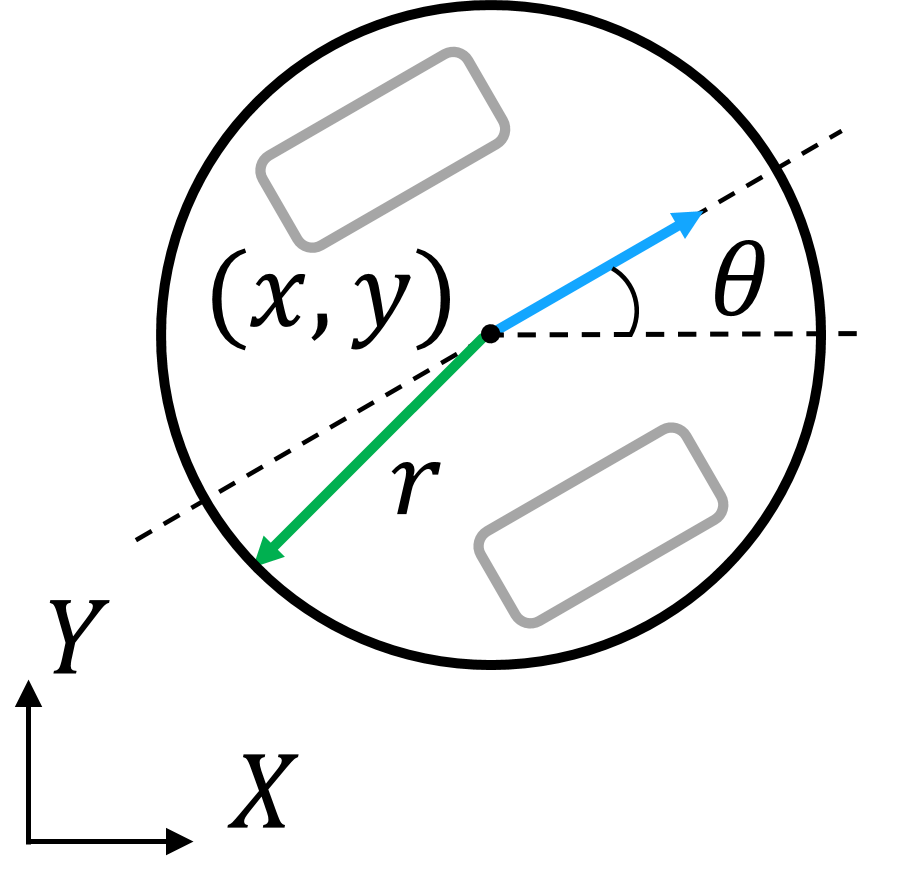}
		\caption{Illustration of a circular differential robot centered at $(x, y)$ with a radius of $r$.}
		\label{fig:differential_drive}
	\end{minipage}
\end{figure*}

This paper focuses on robust MPC-based collision-free reference tracking.
The zero-order robust optimization (zoRO) MPC is applied to track a reference trajectory by a higher-level planner as depicted in Fig.~\ref{fig:control_pipeline}.
The reference trajectory satisfies the system dynamics, 
  but it is not necessarily robustly collision-free.
The robust MPC modifies motion commands when necessary and computes robustified inputs.
The acceleration inputs are then integrated to obtain the velocity commands for the robot.
With regard to obstacles, we assume that they are circular and static.
Full knowledge of the positions and the radii of the obstacles is assumed to be available.
The experiments on a real mobile robot show the efficacy of our approach.

Moreover, we study the suboptimality properties of the zoRO MPC.
For the robots that have a linear input subsystem and are subject to affine constraints on the subsystem, 
  the zoRO algorithm is locally optimal when the robots are sufficiently far away from the obstacles. 
Consequently, the convergence to the reference trajectory is not impaired. 

\subsubsection*{Related Work}
A robust nonlinear MPC framework for lane keeping and obstacle avoidance of ground vehicles is proposed in~\cite{Gao2014}.
The authors compute robust positively invariant (RPI) sets offline
  and tighten the constraints to ensure constraint satisfaction in the presence of disturbances.
Furthermore, many researchers have investigated robustifying collision avoidance in the robust MPC framework utilizing growing tubes~\cite{8796049,8202163}.
The growth of the tubes can be counteracted by incorporating linear feedback laws.
These feedback laws can be optimized~\cite{Messerer2021, NAGY2004411} or precomputed~\cite{MayneTube2011}.
\looseness=-1

\subsubsection*{Notation}
The paper applies to sampled-data systems.
However, for the methods and analysis, we consider a discrete-time setting.
For brevity of notation, for a trajectory $f\colon \mathbb{R}_+\to \mathbb{R}^n$ we write $f_k := f(t_k)$.
An ellipsoidal set $\mathcal{E}$ centered at the origin can be defined by 
  a symmetric positive definite shape matrix $Q$,
$\mathcal{E}(Q)\coloneqq \{y\in\mathbb{R}^n|y^\top Q^{-1}y\leq 1\}$.
A set of natural numbers possibly containing 0 in the interval $\left[a, b \right]\subset\mathbb{N}_0 $ is denoted by $\mathcal{I}_{\left[a, b \right]}$.
The quadratic norm with respect to a symmetric positive definite matrix $Q$ is denoted by $\|x\|^2_Q=x^\top Qx$.

\subsubsection*{Structure}
Section~\ref{sec:problem_statement} states the problem to be addressed in this paper.
Section~\ref{sec:preliminary} briefly recapitulates robust MPC and in particular zero-order robust optimization (zoRO) MPC.
Section~\ref{sec:subopt} studies and numerically validates the optimality of the zoRO MPC in the absence of obstacles.
The implementation details and the real-world experiments are presented in Section~\ref{sec:experiment}.
Section~\ref{sec:conclusion} concludes the paper.

%% file: sections/2-problem.tex
\section{Problem Statement}
\label{sec:problem_statement}

The problem statement consists of three ingredients: uncertain robot system dynamics, collision avoidance, and reference trajectory tracking.

\subsubsection*{Robot system dynamics}
Consider a circular differential robot with radius $r$  (Fig.~\ref{fig:differential_drive}).
The robot system state is 
$$s = \left[ \begin{array}{ccccc}
x & y & \theta & v &\omega
\end{array}\right]^\top \in\mathbb{R}^{n_s},$$ 
with $(x,y)$ denoting the robot center position and $\theta$ the heading angle in some reference inertial frame.
$v$ and $\omega$ denote the forward velocity and angular velocity of the robot.
The system inputs $u\in\mathbb{R}^{n_u}$ are forward acceleration $a$ and angular acceleration $\alpha$:
$$u = \left[ \begin{array}{cc}
a &\alpha
\end{array}\right]^\top \in\mathbb{R}^{n_u}.$$ 
The robot is subject to the following system dynamics:
\begin{equation}
\frac{\mathrm{d}s}{\mathrm{d}t}
= \left[ \begin{array}{ccccc}
v\cos(\theta) & v\sin(\theta) & \omega & a & \alpha
\end{array}\right]^\top
\label{eq:robot_system_dynamics}
\end{equation}
and subject to box constraints on $v$, $\omega$, $a$, $\alpha$ due to actuator limits.
The system is discretized by numerical integration:
$$s_{k+1} = \psi(s_k, u_k)\colon \mathbb{R}^{n_s}\times\mathbb{R}^{n_u}\to \mathbb{R}^{n_s}.$$
The controls within one discretization interval are constant.
Note that the discrete-time system of~\eqref{eq:robot_system_dynamics} is a special case of the more general class of nonlinear dynamics with a linear input subsystem:
\begin{subequations}
	\label{eq:linear_subsystem_whole}
	\begin{align}
	s_k &= \left[ \begin{array}{c}
	s^{\mathrm{kin}}_k \\ s^{\mathrm{lin}}_k
	\end{array}\right],\\
	s^{\mathrm{kin}}_{k+1} &= \psi^{\mathrm{kin}}(s_k),\\
	s^{\mathrm{lin}}_{k+1} &=
	{A}^{\mathrm{lin}} s^{\mathrm{lin}}_k + {B}^{\mathrm{lin}} u_k,
	\label{eq:linear_subsystem}
	\end{align}
\end{subequations}
where 
$s^{\mathrm{kin}}$ is the kinematic part (i.e., $x$, $y$, and $\theta$) and $s^{\mathrm{lin}}$ is the linear part (i.e., $v$ and $\omega$) of the state.
The matrices ${A}^{\mathrm{lin}}\in\mathbb{R}^{n_\mathrm{lin}\times n_\mathrm{lin}}$ and ${B}^{\mathrm{lin}}\in\mathbb{R}^{n_\mathrm{lin} \times n_u}$ are constant.

The theoretical results in this paper are developed for the more general class~\eqref{eq:linear_subsystem_whole}.

The robot motion is subject to process noise:
\begin{equation}
\tilde{s}_{k+1} = \psi(\tilde{s}_k, u_k) + w_k,
\label{eq:process_noise}
\end{equation}
where $\tilde{s}_k$ are the disturbed states and $w_k$ is the process noise.
\yfHighlight{Remark~\ref{remark:w-distribution} discusses the properties of $w_k$ in detail.}

A differential-drive controller is placed in the control pipeline (Fig.~\ref{fig:control_pipeline}).
It tracks the command velocities (i.e., $v^{\mathrm{cmd}}$ and $\omega^{\mathrm{cmd}}$) and 
  is modeled to reduce the deviation of the robot velocities (i.e., $v^{\mathrm{rob}}$ and $\omega^{\mathrm{rob}}$) from the velocity commands:
  \vspace{-10pt}
  \begin{subequations} \label{eq:diff_drive} \begin{align}
  	v^{\mathrm{rob}}_{k+1} &= v^{\mathrm{cmd}}_k + (v^{\mathrm{rob}}_k-v^{\mathrm{cmd}}_k)e^{\left(-{\Delta t}/{\tau}\right)},\\
  	\omega^{\mathrm{rob}}_{k+1} &= \omega^{\mathrm{cmd}}_k + (\omega^{\mathrm{rob}}_k-\omega^{\mathrm{cmd}}_k)e^{\left(-{\Delta t}/{\tau}\right)},
  	\end{align}
\end{subequations}
where $\tau\in\mathbb{R}_{> 0}$ describes the convergence rate attained by the differential-drive controller and $\Delta t\in\mathbb{R}_{> 0}$ is the discretization interval.

\subsubsection*{Collision avoidance}
The robot moves through a set of \emph{static} circular obstacles $\{o_1, o_2, \cdots o_{n_o}\}$.
Obstacle $o_j$ has a radius of $r_j^{\mathrm{obs}}$ and is centered at $(x_j^{\mathrm{obs}}, y_j^{\mathrm{obs}})$.
The exact obstacle positions and radii are considered to be known. 
The Cartesian coordinates $x_j^{\mathrm{obs}}$ and $y_j^{\mathrm{obs}}$ are in the same inertial frame as the robot pose.
We define a collision-free trajectory as one where the robot does not collide with any obstacle at any time instant. 
We enforce collision avoidance constraints at $t_k$ for all $k\in\mathcal{I}_{\left[0, N\right]}$.
For all $j\in\mathcal{I}_{\left[1,n_o  \right] }$, we impose
\begin{equation}
\sqrt{{(x_k\! -\! x_j^{\mathrm{obs}})}^2 \!+\!{(y_k\! -\! y_j^{\mathrm{obs}})}^2} \geq r_j^{\mathrm{obs}} + r.
\label{eq:nominal_collision_cstr}
\end{equation}

The set of states that satisfy~\eqref{eq:nominal_collision_cstr} and the velocity constraints is denoted as $\mathcal{S}\subset \mathbb{R}^{n_s}$.
The set of inputs that satisfy the acceleration constraints is denoted as $\mathcal{U}\subset \mathbb{R}^{n_u}$.

\subsubsection*{Reference trajectory tracking}
We assume that the state and input reference trajectories are given:
\begin{equation}\label{eq:reference_trajectory}
		s^{\mathrm{ref}}(t)\colon\mathbb{R}_+\to \mathbb{R}^{n_s}~\text{and}~
		u^{\mathrm{ref}}(t)\colon\mathbb{R}_+\to \mathbb{R}^{n_u}.
\end{equation}
The reference trajectories are furthermore assumed to be consistent with the system dynamics~\eqref{eq:robot_system_dynamics}.
The stage and terminal costs of the optimal control problem (OCP), assuming no disturbance, are the sum of the weighted squared tracking error over the prediction horizon:
\begin{subequations}
	\label{eq:stage_cost}
	\begin{align}
	l_k(s_k, u_k;t)\coloneqq 
	  & \|u_k\!-\!u^{\mathrm{ref}}_{t+k}\|_R^2 + \|s_k\!-\!s^{\mathrm{ref}}_{t+k}\|_Q^2, \\
	l_f(s_N;t)\coloneqq 
	& \|s_N\!-\!s^{\mathrm{ref}}_{t+N}\|_{Q_e}^2.
	\end{align}
\end{subequations}
For compactness of notation, we define $z \coloneqq \left[s_0, u_0, \cdots, s_{N-1}, u_{N-1}, s_N\right]$.
Given the current robot state $\bar{s}_{0|t}$,
  the optimal value function is
\begin{mini}|s|
	{z}{l_f(s_N;t) + \sum_{k=0}^{N-1} l_k(s_k, u_k;t)}
	{\label{eq:opt_obj_func}}{V_{N}(\bar{s}_{0|t}, t)\coloneqq}
	\addConstraint{s_0}{\! = \! \bar{s}_{0|t}}{}
	\addConstraint{s_{k+1}}{\! = \! \psi(s_k, u_k),~}{k\in\mathcal{I}_{\left[0, N-1 \right]}}
	\addConstraint{s_k}{\! \in \! \mathcal{S},}{k\in\mathcal{I}_{\left[0, N \right]}}
	\addConstraint{u_k}{\! \in \! \mathcal{U},}{k\in\mathcal{I}_{\left[0, N-1 \right]}}.
\end{mini}

\begin{aim}
	\label{porb}
	Given a mobile robot~\eqref{eq:robot_system_dynamics} with the control pipeline (Fig.~\ref{fig:control_pipeline})
	  and a reference trajectory~\eqref{eq:reference_trajectory},
	  design a robust MPC-based controller that achieves the following:
	\begin{enumerate}[\hspace{3mm}(a)]
		\item Robust collision avoidance: 
		  The robot motion satisfies collision avoidance constraints~\eqref{eq:nominal_collision_cstr} in the presence of bounded 
		  process noise~\eqref{eq:process_noise}.
		\item Convergence to the reference trajectory in the absence of obstacles.\label{prob:convergence}
		\item Real-time feasible computation times on typical robot computing platforms.
	\end{enumerate}
\end{aim}

%% file: sections/3-preliminary.tex
\section{Preliminary}
In this section, we briefly review the robust MPC formulation and the zoRO method.
\label{sec:preliminary}

\subsection{Robust MPC with Ellipsoidal Sets}
\label{sec:preliminary_RMPC}

Robust MPC aims to robustify the optimal control for systems subject to bounded disturbance.
Following~\cite{Houska2011}, we approximate the disturbance bounds by ellipsoidal sets $\mathcal{E}(\Sigma)$.
Let $\Sigma_0 = \bar{\Sigma}_{0|t}$.
The disturbance ellipsoidal tubes are propagated using the recursion:
\begin{equation}
\begin{split}
\Sigma_{k+1} &= (A_k\! -\! B_kK)  \Sigma_k {(A_k\! -\! B_kK)}^\top + W\\
    &\coloneqq \Phi(\Sigma_k, s_k, u_k, W), \ k\in\mathcal{I}_{\left[0 , N-1 \right]},
\end{split}
\label{eq:dist_prop}
\end{equation}
where $A_k\coloneqq\frac{\partial \psi(s_k, u_k)}{\partial s_k}$ and $B_k\coloneqq\frac{\partial \psi(s_k, u_k)}{\partial u_k}$ and $W\in\mathbb{R}^{n_s \times n_s}$ is positive definite.
$K$ is a precomputed feedback gain matrix 
  that counteracts the growth of ellipsoidal tubes.

\begin{remark}
	\label{remark:w-distribution}
    Equation~\eqref{eq:dist_prop} is akin to propagation of covariance matrices, e.g., in chance-constrained MPC~\cite{feng2020inexact, Hewing2017}.
    In robust MPC, the propagation of the ellipsoidal tubes~\eqref{eq:dist_prop}
    is valid for linear systems when the entire disturbance trajectory lies in a higher dimensional ellipsoid:
    $$\left[\begin{array}{cccc}
    	w_0^\top&w_1^\top& \cdots & w_{N-1}^\top
    \end{array}  \right]^\top \in\mathcal{E}\left(\text{diag}\Big(\underbrace{W, \cdots, W}_N \Big) \right).$$
	Note that this assumption on the disturbance is different from that the disturbance at every time instant $w_k$ in a bounded set 
	    as in, e.g.,~\cite{MayneTube2011, Kohler2018, MAYNE2005219}.
\end{remark}

Consider the following robust optimal control problem with ellipsoidal sets:
\begin{mini}|s|
	{\substack{z,\\  \Sigma_{0}, \dots, \Sigma_{N}}}
	{l_f(s_N) + \sum_{k=0}^{N-1} l(u_k, s_k)}
	{\label{eq:robustMPC}}{}
	\addConstraint{s_0}{=\bar{s}_{0|t}}{}
	\addConstraint{\Sigma_0}{= \bar{\Sigma}_{0|t}}{}
	\addConstraint{s_{k+1}}{=\psi(s_k, u_k),}{k\in\mathcal{I}_{\left[0, N-1 \right]}}
	\addConstraint{\Sigma_{k+1}}{=\Phi(\Sigma_k, s_k, u_k, W), }{k\in\mathcal{I}_{\left[0, N-1 \right]}}
	\addConstraint{0}{\geq \! h_k(s_k, u_k) \!+\! \beta_k(\Sigma_k, s_k, u_k),~}{k\in\mathcal{I}_{\left[0, N-1 \right]}}
	\addConstraint{0}{\geq \! h_N(s_N) \! +\!  \beta_N(\Sigma_N, s_N)},
\end{mini}
where $h_k\colon \mathbb{R}^{n_s} \times \mathbb{R}^{n_u} \rightarrow \mathbb{R}^{n_{h_k}}$ are the stage constraints
and $h_N\colon \mathbb{R}^{n_s} \rightarrow \mathbb{R}^{n_{h_N}}$ is the terminal constraint.
The backoff terms $\beta \in \mathbb{R}^{n_{h_k}}$ account for the disturbances and keep the system away from the constraints~\cite[Section 3.2.2]{Gillis2015}:
\begin{equation}
\begin{split}
\beta_{k} &\coloneqq 
\sqrt{\nabla h_{k}(\cdot, \cdot)^\top 
	\left[\begin{array}{c}
	I\\ K
	\end{array} \right] \Sigma_k
	{\left[\begin{array}{c}
		I\\ K
		\end{array} \right]}^\top\nabla h_{k}(\cdot, \cdot)},
\\[1em]
\beta_{N} &\coloneqq \sqrt{\nabla h_{N}(s_N)^\top
	\Sigma_N
	\nabla h_{N}(s_N)}.
\end{split}
\label{eq:backoff_term}
\end{equation}

\begin{remark}
An alternative to implement robust MPC is to consider circular tubes~\cite{Kohler2018, WABERSICH2021109597}.
This approach bounds the disturbance by hyperspheres with radii $\epsilon$.
The system capability of counteracting against disturbances is represented by a scalar $\rho$ such that
$\epsilon_{k+1} = \rho \epsilon_k+ \epsilon$.
This approach is particularly attractive in case the disturbance set is well approximated by hyperspheres,
  but otherwise suffers from significant conservativeness due to large backoff terms.
\end{remark}

\subsection{Zero-order Robust Optimization (zoRO)}
Robust MPC problems~\eqref{eq:robustMPC} can be solved with standard OCP solvers
  by augmenting the system states with $\Sigma_k$.
However, this leads to a quadratic increase in the system variables since $\Sigma_{k}\in\mathbb{R}^{n_s\times n_s}$. 
The zoRO method~\cite{Zanelli2021a} reduces the computational complexity by iteratively solving optimization subproblems:
\begin{mini}|s|
	{z}{l_f(s_N) + \sum_{k=0}^{N-1} l_k(u_k, s_k)}{\label{eq:robustOCP_zeroOrder}}{}
	\addConstraint{s_0}{=\bar{s}_{0|t}}{}
	\addConstraint{s_{k+1}}{=\psi(s_k, u_k),}{k\in\mathcal{I}_{\left[0, N-1 \right]}}
	\addConstraint{0}{\geq h_k(x_k, u_k) + \hat{\beta}_k,~}{k\in\mathcal{I}_{\left[0, N-1 \right]}}
	\addConstraint{0}{\geq h_N(x_N) + \hat{\beta}_N}{},
\end{mini}
where the ellipsoidal sets are eliminated from the optimization variables and the backoff terms are fixed~\cite{Zanelli2021a}.
After one subproblem is solved, the ellipsoidal sets are updated using~\eqref{eq:dist_prop}.

\begin{remark}
	The optimization problem~\eqref{eq:robustOCP_zeroOrder} is a nonlinear programming (NLP) problem.
	It can for instance be solved using an interior-point (IP) method~\cite{Wright1997}
	  as in the numerical tests in Section~\ref{sec:opt_wo_obs}.
	An alternative approach of solving~\eqref{eq:robustOCP_zeroOrder} is to use sequential quadratic programming (SQP) and solve only one quadratic programming (QP) problem before updating the backoff terms, as in the original zoRO paper~\cite{Zanelli2021a}.
	Whether to and how often to update the backoff terms are choices of implementation.
	In~\cite{Hewing2017}, the ellipsoids are evaluated and constraints are tightened based only on the previous MPC solutions.
	The implementation details of the zoRO when running on the real robot are presented in Section~\ref{sec:implementation}.
\end{remark}

%% file: sections/4-zoMPC.tex
\section{Optimality and Suboptimality of zoRO}
\label{sec:subopt}

In general, zoRO converges to a suboptimal point since the backoff terms are fixed and their gradients are neglected~\cite{Zanelli2021a}.
In this section we will show however that, for our use case,
  zoRO converges to a local optimum of
  \eqref{eq:robustMPC} in relevant conditions.
In particular,
  this is the case 
  when the robot is sufficiently distant from other obstacles, 
     i.e., the uncertain predicted trajectory does not overlap with any obstacle.
  Convergence properties can then be inherited from other robust MPC schemes.

\subsection{Suboptimality}
In zoRO, the shape matrices $\Sigma$ of the ellipsoidal sets are eliminated from the optimization variables.
Fixing the ellipsoidal sets and backoff terms leads to inexact first-order optimality conditions. 
The gradients of the backoff terms with respect to the system states and control inputs are disregarded.
Define $g_k(z)\coloneqq \Sigma_k$, where $\Sigma_k$ depends on $z$ and $W$. 
The robustified constraints in this notation are
  $h_k(z) + \beta_k(g_k(z), z) \leq 0$.
The disregarded term from the gradient of the Lagrangian is
\begin{equation}
	\frac{\partial \mathcal{L}_{\eqref{eq:robustMPC}}}{\partial z} - \frac{\partial \mathcal{L}_{\eqref{eq:robustOCP_zeroOrder}}}{\partial z} = \sum_{k=0}^N {\left(\frac{\partial \beta_k}{\partial z} + \frac{\partial\beta_k}{\partial g_k}\frac{\partial g_k}{\partial z}\right)}^\top\mu_k,
	\label{eq:deviation_Jacobian}
\end{equation}
where $\mathcal{L}_{\eqref{eq:robustMPC}}$ and $\mathcal{L}_{\eqref{eq:robustOCP_zeroOrder}}$ denote the Lagrangians of
  \eqref{eq:robustMPC} and \eqref{eq:robustOCP_zeroOrder}, respectively,
  and $\mu_k$ the Lagrange multiplier of the inequality constraints.
For more details about the resulting suboptimality,
  see~\cite{Zanelli2021a}.

\subsection{Optimality in the Absence of Obstacles}
\label{sec:opt_wo_obs}
The mobile robot system~\eqref{eq:robot_system_dynamics} contains a linear input subsystem~\eqref{eq:linear_subsystem_whole}.
Due to the linear input subsystem, 
  zoRO MPC is locally optimal
    when collision avoidance constraints are inactive.
This feature can be generalized to other wheeled robots that can be written to have a linear input subsystem.

\begin{assumption}
\label{asm:linear_subsystem}
The robot system is of the form~\eqref{eq:linear_subsystem_whole}.
The linear feedback law only acts on the state components in the linear input subsystem:
\begin{equation}
\label{eq: K_Matrix}
K = \left[ 
\begin{array}{cc}
0^{n_u \times n_{\mathrm{kin}}} & K^{\mathrm{lin}}
\end{array} \right],
\end{equation}
  where $K^{\mathrm{lin}}\in\mathbb{R}^{n_u\times n_{ \mathrm{lin}}}$ is constant.
\yfHighlight{The OCP is subject to affine constraints $h_k^{\mathrm{aff}}$ on the linear input subsystem 
  and nonlinear collision avoidance constraints $h_k^{\mathrm{coll}}$}:
$$h_k(s_k, u_k) = \left[ \begin{array}{c}
h_k^{\mathrm{aff}}(s_k^{\mathrm{lin}}, u_k)\\
h_k^{\mathrm{coll}}(s_k^{\mathrm{kin}})
\end{array}\right]\leq 0.$$
\end{assumption}

\begin{theorem}
	Let Assumption~\ref{asm:linear_subsystem} hold.
	If collision avoidance constraints are inactive, 
	  zoRO converges to a solution of~\eqref{eq:robustMPC}.
	\label{thm:zo_exact_same}
\end{theorem}
\begin{proof}
	Partition the ellipsoidal matrix $\Sigma_{k}$ as
	$$ \Sigma_k = \left[ \begin{array}{cc}
	\Sigma_{k}^{\mathrm{kin}} & \Sigma_{k}^{\mathrm{cpl}}\\
	\left( \Sigma_{k}^{\mathrm{cpl}}\right)^\top  & \Sigma_{k}^{\mathrm{lin}}
	\end{array}\right], $$
	where $\Sigma_{k}^{\mathrm{kin}}$ is the disturbance matrix of $s_k^{\mathrm{kin}}$
	and $\Sigma_{k}^{\mathrm{lin}}$ is that of $s_k^{\mathrm{lin}}$
	and $\Sigma_{k}^{\mathrm{cpl}}$ represents the coupling part between $s_{k}^{\mathrm{kin}}$ and $s_{k}^{\mathrm{lin}}$.
	Similarly,
	  partition the disturbance matrix $W$ and
	  \yfHighlight{we have $W^{\mathrm{lin}}$ which describes the one-step disturbance of $s^{\mathrm{lin}}$}.
	Subsequently, we partition ${A}_k$ in~\eqref{eq:dist_prop} into
	$\left[ \begin{array}{cc}
	{A}_{k}^{\mathrm{kin}} & {A}_{k}^{\mathrm{cpl}}\\
	0 & {A}^{\mathrm{lin}}
	\end{array}\right]$
	and $B_k$ into 
	$\left[ \begin{array}{c}
	{B}_{k}^{\mathrm{kin}}\\
	{B}^{\mathrm{lin}}
	\end{array}\right]$.
	Given the assumptions on the feedback gain $K$ matrix~\eqref{eq: K_Matrix},
	it is easily verifiable that
	\begin{equation*}
	A_k - B_k K = \left[ \begin{array}{cc}
	{A}_{k}^{\mathrm{kin}} & {A}_{k}^{\mathrm{cpl}} - {B}_{k}^{\mathrm{kin}}K^{\mathrm{lin}}\\
	0 & {A}^{\mathrm{lin}} - B^{\mathrm{lin}}K^{\mathrm{lin}}
	\end{array}\right].
	\end{equation*}
	For brevity of notation, we denote 
	$\hat{A}^{\mathrm{lin}} \coloneqq {A}^{\mathrm{lin}} - B^{\mathrm{lin}}K^{\mathrm{lin}}$
	such that
	$$\Sigma_{k+1}^{\mathrm{lin}} = \hat{A}^{\mathrm{lin}}\Sigma_{k}^{\mathrm{lin}}\left(\hat{A}^{\mathrm{lin}}\right)^\top + W^{\mathrm{lin}},$$
	which is independent of the optimization variables $s_k$ and $u_k$.
	As the constraints on the linear input subsystem (i.e., $h_k^{\mathrm{aff}}(s_k^{\mathrm{lin}}, u_k)\leq 0$) are affine,
	  we can conclude that the gradients of the backoff terms of these constraints over the states and control inputs are zero. 
	
	Moreover, since the collision avoidance constraints (i.e., 
	$h_k^{\mathrm{coll}}(s_k^{\mathrm{kin}})\leq 0, \forall k\in\mathcal{I}_{[0, N]}$) are inactive,
	  the corresponding Lagrange multipliers are zero.
	Therefore, the disregarded gradients~\eqref{eq:deviation_Jacobian} of both affine constraints and collision avoidance constraints are zero.
	The subproblem \eqref{eq:robustOCP_zeroOrder} has the same first-order optimality conditions in terms of the states and control inputs as \eqref{eq:robustMPC}. 
\end{proof}

\begin{figure*}[tb]
	\centering
	\vspace{2pt}
	\begin{minipage}[t]{0.34\textwidth}
		\centering
		\includegraphics[width=0.82\columnwidth]{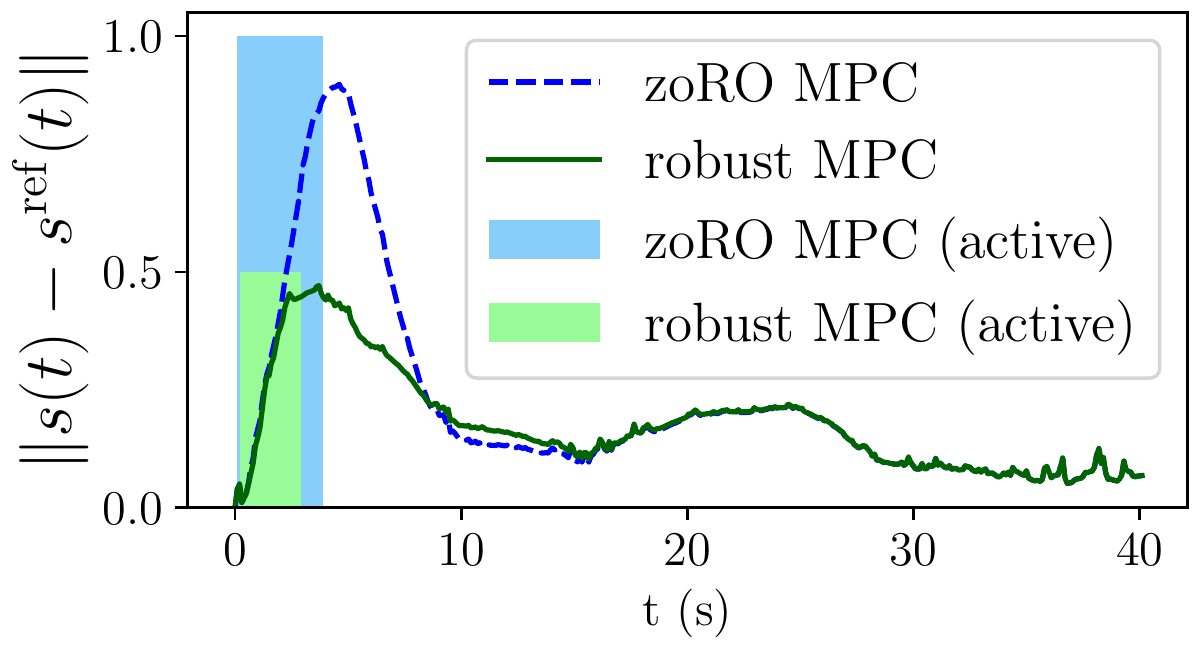}
		\caption{
			Tracking errors of the robust MPC~\eqref{eq:robustMPC} and those of zoRO~\eqref{eq:robustOCP_zeroOrder}.
			The boxes indicate that at least one collision avoidance constraint is active for each of the MPC approaches.
		}
		\label{fig:res_convergence}
	\end{minipage}
	\hspace{0.02\textwidth}
	\begin{minipage}[t]{0.34\textwidth}
		\centering
		\includegraphics[width=\columnwidth]{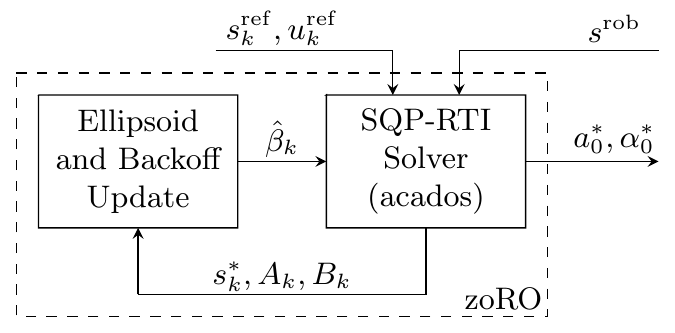}
		\caption{\yfHighlight{Diagram of the zoRO algorithm. 
			The updated backoff terms $\hat{\beta}_k$ are passed to the SQP-RTI solver, 
			which performs one QP iteration based on the robot state $s^{\mathrm{rob}}$ and the reference trajectories $s^{\mathrm{ref}}_k$ and $u^{\mathrm{ref}}_k$.}}
		\label{fig:diagram_RMPC}
	\end{minipage}
	\hspace{0.02\textwidth}
	\begin{minipage}[t]{0.25\textwidth}
		\centering
		\includegraphics[width=0.9\columnwidth]{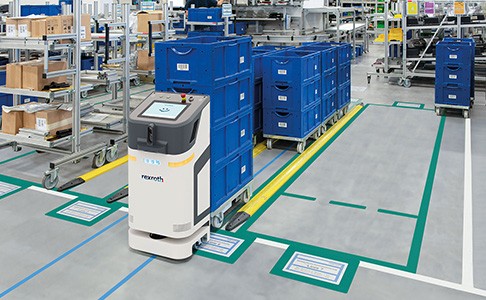}
		\caption{The ActiveShuttle AMR by Bosch Rexroth. The top-view robot shape is outer-approximated by a circle.}
		\label{fig:BoschActiveShuttle}
	\end{minipage}
\end{figure*}

To numerically validate Theorem~\ref{thm:zo_exact_same}, 
  we solve Problem~\ref{porb} with the robust MPC~\eqref{eq:robustMPC} and zoRO MPC~\eqref{eq:robustOCP_zeroOrder}
  by using the primal-dual interior point optimizer (IPOPT)~\cite{Waechter2005} interfaced through CasADi~\cite{Andersson2018}.
\yfHighlight{When collision avoidance constraints are inactive,
  the difference between the first control input is below $10^{-7}$, which can be interpreted as numerical noise.}

\begin{corollary}
	\label{cor:same_convergence}
	Let Assumption~\ref{asm:linear_subsystem} hold and collision avoidance constraints be inactive.
	The robust trajectory tracking problem solved by the zoRO
	  MPC~\eqref{eq:robustOCP_zeroOrder} has the same property of convergence to the reference trajectories as that solved by the robust MPC~\eqref{eq:robustMPC}.
\end{corollary}

The proof of convergence to reachable dynamic reference trajectories 
    under the closed-loop MPC 
    for nonlinear systems in the absence of disturbance can be found in~\cite[Theorem~2]{Kohler2019}.
It is proven in~\cite[Theorem~8]{Kohler2018} that nonlinear systems subject to bounded disturbance can be stabilized at the origin if the system is locally incrementally stabilizable.

Here we show that the robot trajectories robustly converge to the reference trajectories in practice.
In simulation, we drive the robot along a reference trajectory with the closed-loop control of the robust MPC~\eqref{eq:robustMPC} and the zoRO MPC~\eqref{eq:robustOCP_zeroOrder}.
After the robot passes by the obstacle, the tracking error reduces, gets into a bounded range, and stays within the range (Fig.~\ref{fig:res_convergence}).
Moreover, it is remarkable that indeed the tracking errors of the robust MPC and the zoRO MPC are almost the same 
  when the collision avoidance constraints are inactive.
Corollary~\ref{cor:same_convergence} is supported by numerical tests.

%% file: sections/5-experiments.tex
\section{Implementation and Real-world Experiment Results}
\label{sec:experiment}

In this section we first discuss some implementation details of the zoRO MPC.
Subsequently,
  we evaluate results from real-world experiments
  in terms of robustness, planning efficiency, and computational efficiency.

\subsection{Implementation}
\label{sec:implementation}
\subsubsection*{Linear Feedback Gain Matrix}
As mentioned in Section~\ref{sec:preliminary_RMPC},
  linear disturbance feedback laws are often incorporated in the robust MPC framework to counteract the growth of the tubes.
In our implementation, 
  we take the differential drive controller in Fig.~\ref{fig:control_pipeline}
    as the disturbance feedback in the robust MPC.
Given the model of the differential drive controller~\eqref{eq:diff_drive},
  we compute the constant feedback gain matrix $K$ such that
\begin{equation*}
B^{\mathrm{lin}}K^{\mathrm{lin}} = \left[ 
\begin{array}{cc}
1-\exp\big(-{\Delta t}/{\tau}\big)&0\\
0&1-\exp\big(-{\Delta t}/{\tau}\big)
\end{array}
\right].
\end{equation*}
We consider that the feasible set of accelerations of the differential-drive robot is strictly and sufficiently larger than
  the bounds considered in the robust MPC formulation to allow omitting the backoff terms on the accelerations.
This latter aspect is driven by practical considerations of our experimental platform and admittedly a source of
  conservativeness.
\looseness=-1
  
\subsubsection*{Zero-order robust optimization (Fig.~\ref{fig:diagram_RMPC})}

At every sampling time, the current robot state and the reference trajectories are passed to zoRO.
We first compute the fixed backoff terms
  based on the solution obtained at the previous sampling time. 
After solving~\eqref{eq:robustOCP_zeroOrder} with one QP iteration,
  we update the backoff terms using the latest solutions $s^*_k$ and the corresponding \yfHighlight{Jacobian matrices} $A_k, B_k$.
The updated problem is solved with another QP iteration and 
  the zoRO outputs the first optimal control inputs $a^*_0, \alpha^*_0$.
The optimization problem is solved using the SQP real-time iteration (RTI) in acados~\cite{Verschueren2021} and the QP solver HPIPM~\cite{Frison2020a}.
To compensate for the computational delay,
  we set the initial robot state $\bar{s}_0$ to be the simulated state 
  of the robot at the time we expect the currently computed input to be applied.
The measurement noise is not in the scope of the discussion and it is assumed to be zero. 
The initial disturbance matrix $\bar{\Sigma}_{0|t}$ is chosen 
  to bound the system process noise in the duration of the computational delay.

\subsection{Hardware and Software Setup}

\begin{figure*}[tb]
	\centering
	\vspace{0.07in}
	\begin{subfigure}[b]{0.45\columnwidth}
		\includegraphics[width=1\linewidth]{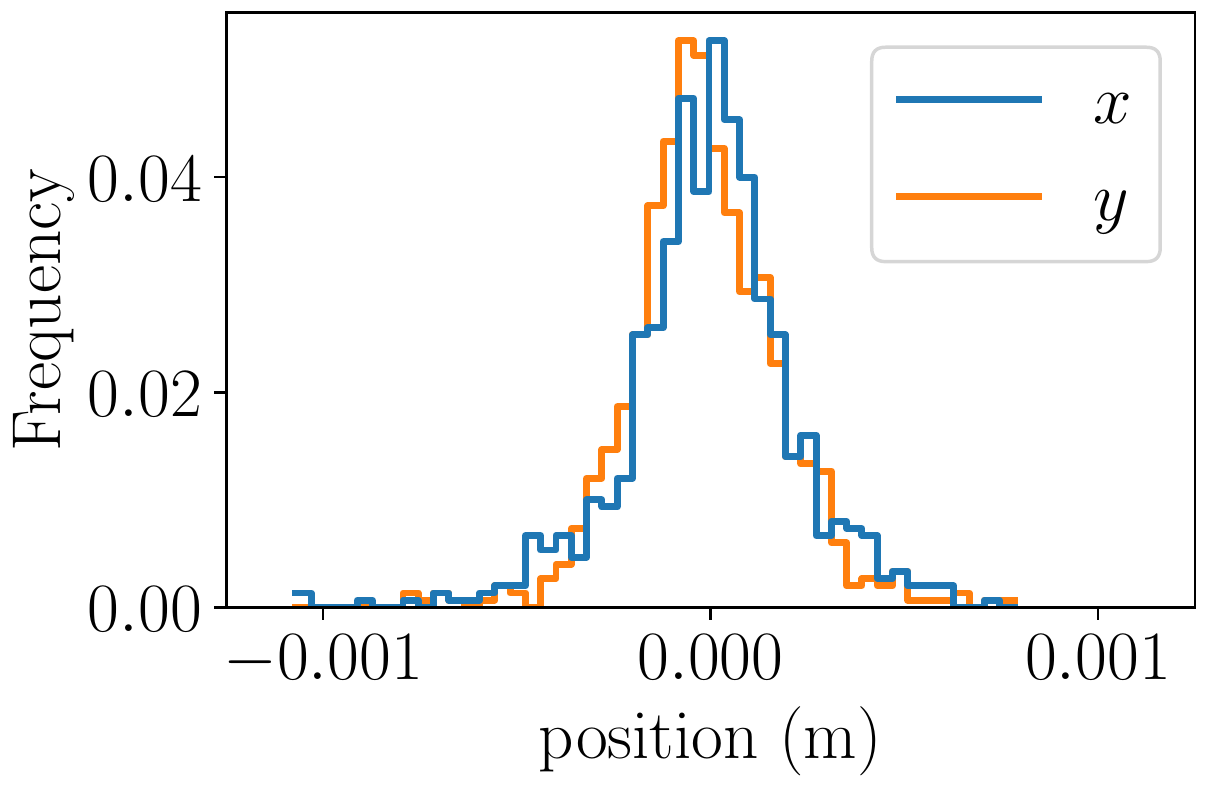}
	\end{subfigure}
    \hspace{0.02\textwidth}
	\begin{subfigure}[b]{0.45\columnwidth}
		\includegraphics[width=1\linewidth]{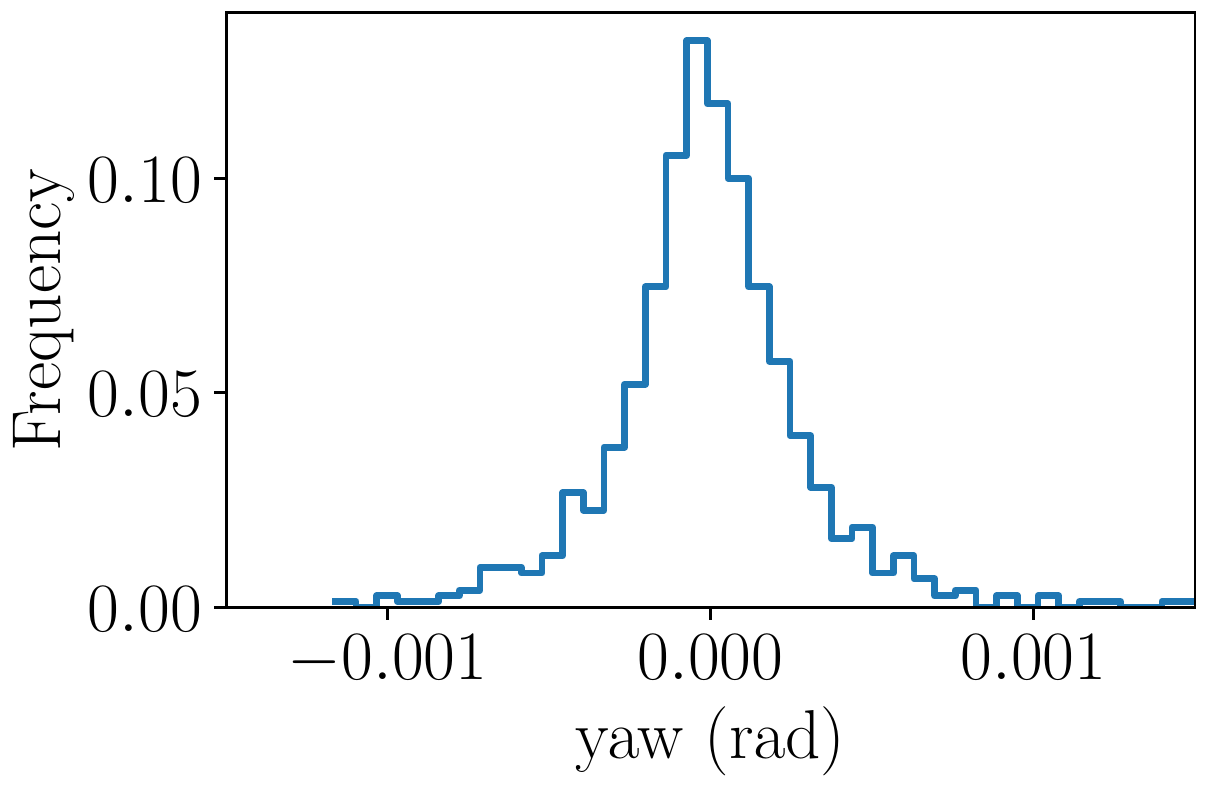}
	\end{subfigure}
	\hspace{0.02\textwidth}
	\begin{subfigure}[b]{0.45\columnwidth}
		\includegraphics[width=1\linewidth]{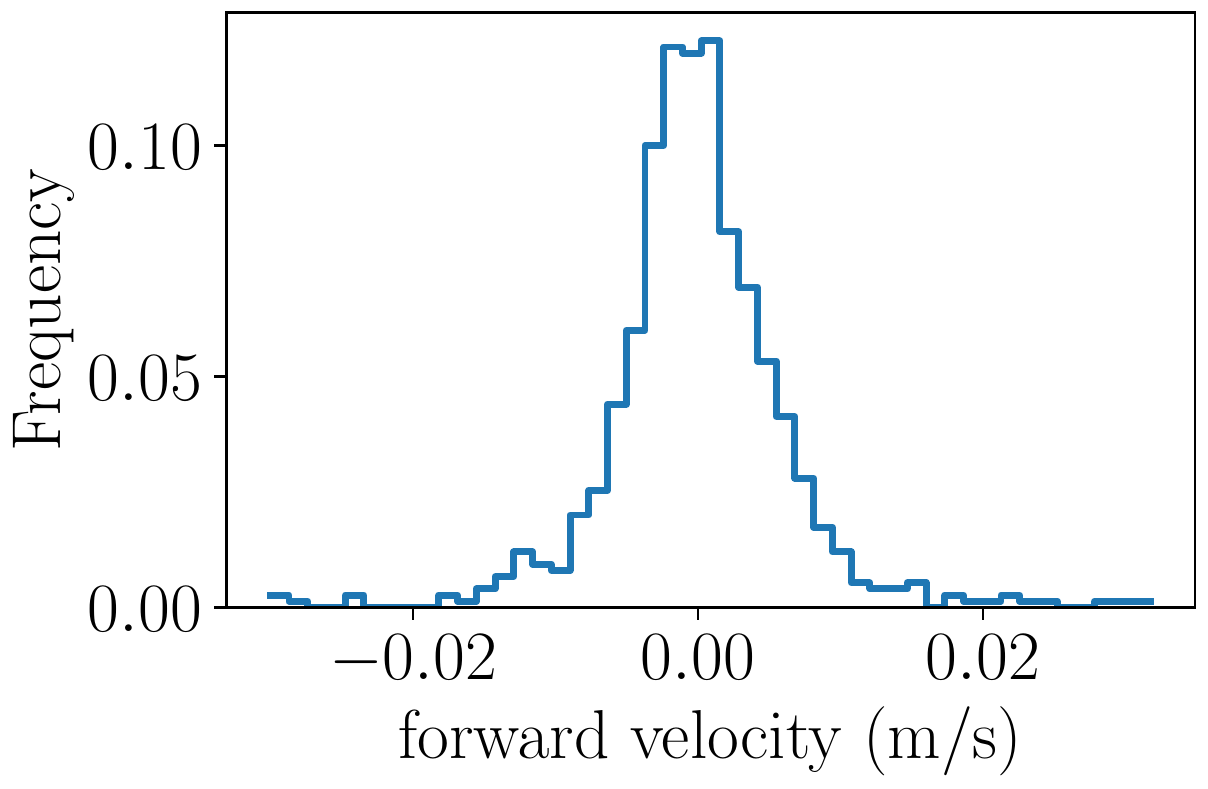}
	\end{subfigure}
    \hspace{0.02\textwidth}
	\begin{subfigure}[b]{0.45\columnwidth}
		\includegraphics[width=1\linewidth]{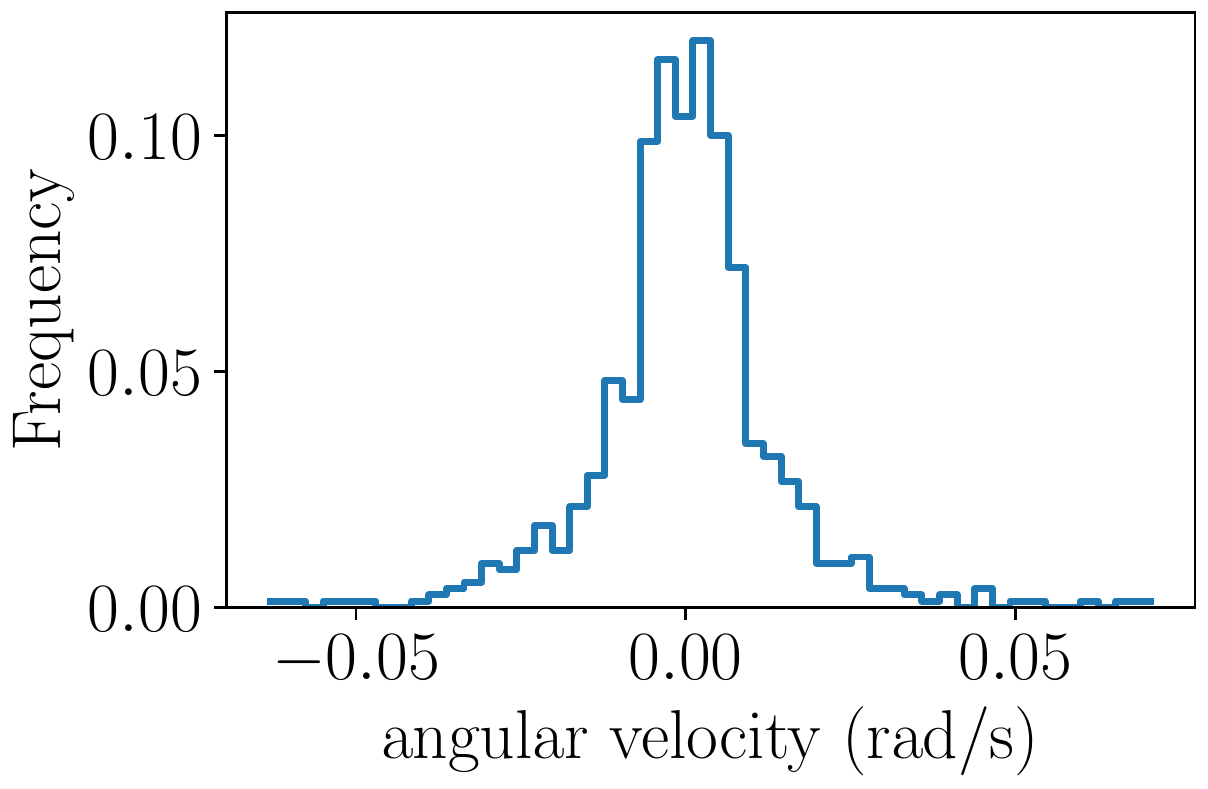}
	\end{subfigure}
	\caption{The process noise measured with the motion capture system.
		\yfHighlight{Note that the process noise on $v$ and $\omega$ is one magnitude larger than on $x$, $y$, and $\theta$.}
	}
	\label{fig:measured_dist}
\end{figure*}

\begin{figure}[tb]
	\centering
	\begin{minipage}[t]{0.23\textwidth}
		\centering
		\includegraphics[width=\columnwidth]{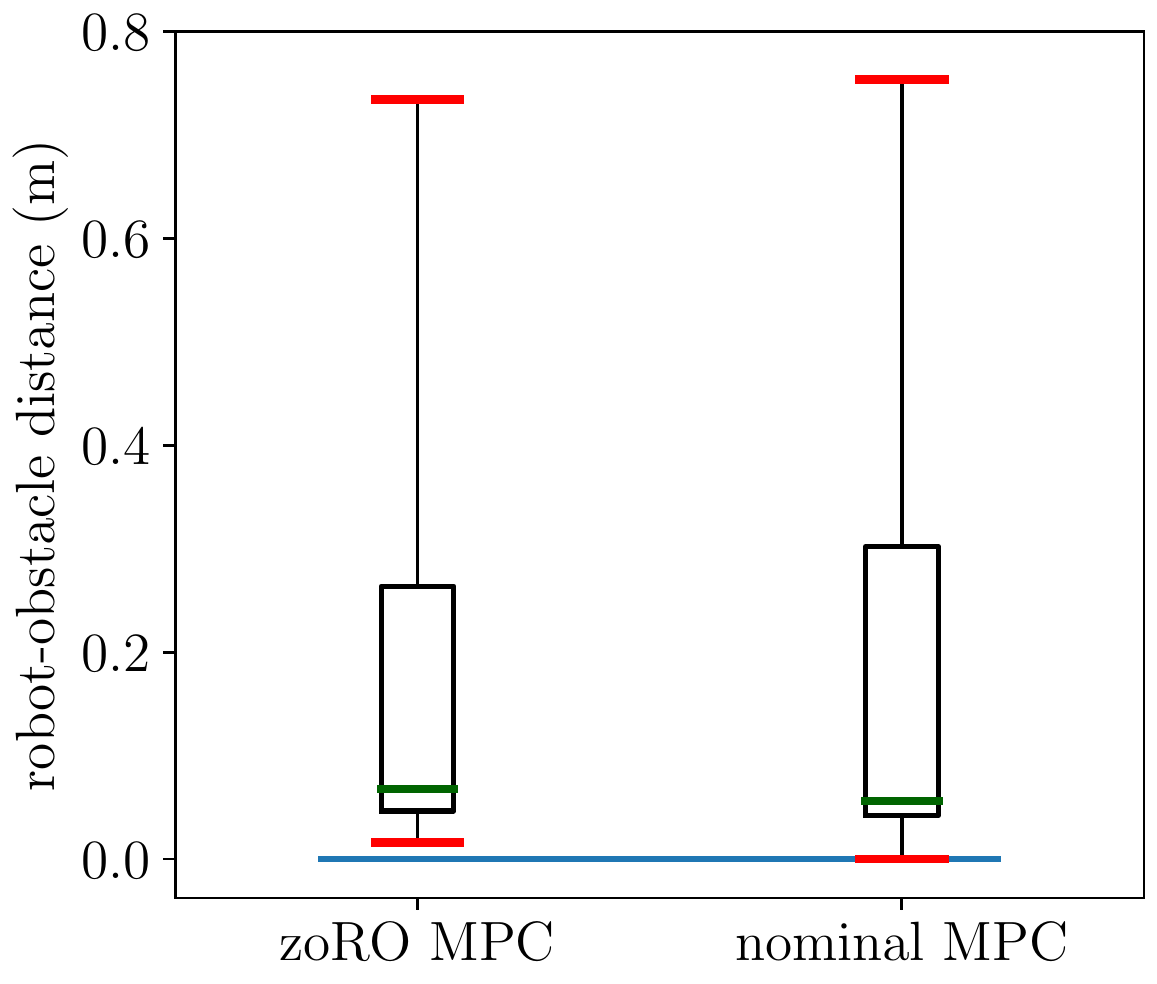}
		\caption{Distances between robot and obstacles while some collision avoidance constraints are active.
		}
		\label{fig:min_distance2obs}
	\end{minipage}
    \hspace{0.01\textwidth}
	\begin{minipage}[t]{0.23\textwidth}
		\centering
		\includegraphics[width=0.96\columnwidth]{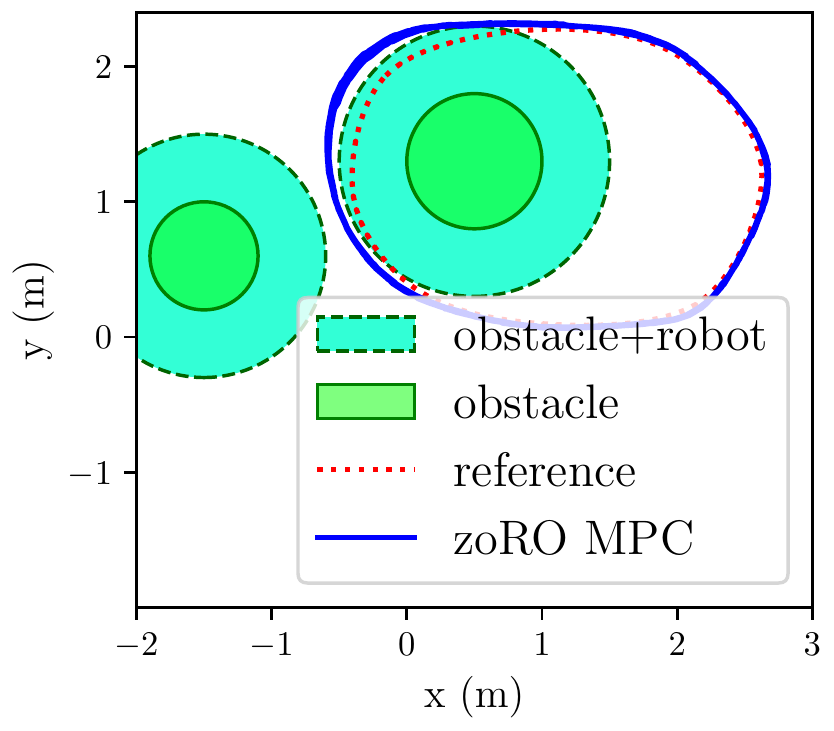}
		\caption{Robot trajectory \yfHighlight{(depicting nine cycles)} and obstacles.}
		\label{fig:ellipsoid_vs_scalar}
	\end{minipage}
\end{figure}

The experiment is carried out on a differential-drive ActiveShuttle by Bosch Rexroth (Fig. \ref{fig:BoschActiveShuttle}).
The robot shape is outer-approximated by a circle of radius 0.5~meters.
The ActiveShuttle is equipped with a quad-core Intel i5-7300U CPU and 8 GB RAM.
We run ROS2~\cite{scirobotics_abm6074} and Nav2~\cite{macenski2020marathon2} on the robot. 
The collision-free zoRO MPC is called every 50~milliseconds.

\subsubsection*{Reference Trajectory}
We consider a closed path along which a Nav2 controller guides the robot.
Two virtual obstacles are placed close to the path.
The Nav2 controller (a.k.a. local planner) is unaware of them.
Consequently, the robot would crash into these virtual obstacles
  if it followed the Nav2 control commands.
We take these `unsafe' trajectories generated by the Nav2 controller as reference trajectories.
Note that the reference trajectories continuously change over time in this particular setup.

\subsection{Bounds of Process Noise}
\label{sec:dist_est}

We estimate the bounds of one-step process noise on the real robot.
First, we get the robot 3D poses by the OptiTrack motion capture system and project the tracked 3D poses onto the 2D plane.
Then, we estimate the linear and angular velocities by solving a least-squares problem based on the robot kinematics.
Taking the tracked 2D poses and the computed velocities as the ground truth,
  we measure the process noise by taking the difference between the numerical integration $\psi(s_k, u_k)$ and the robot state at time $t_{k+1}$.
The measured process noise is plotted in Fig.~\ref{fig:measured_dist}.
It can be observed that the process noise in robot positions and headings is on a much smaller scale than that in forward and angular velocities.
This is because the linear input subsystem is prone to disturbance including computational delay and external forces acting on the mobile robot.
The values at three standard deviations are set as the bounds of process noise.
\looseness=-1

\subsection{Experimental Results}
\subsubsection{Robustness}
To test the system robustness,
  we record the minimal distance of the mobile robot to the virtual obstacles when collision avoidance constraints are active (Fig.~\ref{fig:min_distance2obs}).
\yfHighlight{The body of the boxplot contains the recorded data between the lower and the upper quartile
  and the green line represents the median.}
As shown in the figure, the minimal distance when running the zoRO MPC is slightly over zero.
In contrast, the nominal MPC results in a distance slightly less than zero.
\yfHighlight{That means that the nominal MPC led to a collision due to process noise.}

\subsubsection{Robot Motion Efficiency}

In order to evaluate the robot motion efficiency,
  we place two obstacles on the two sides of the reference path.
As ellipsoidal sets tightly bound the potential disturbance,
  the backoff terms for collision avoidance are not unnecessarily big.
Thus, the zoRO MPC with ellipsoidal sets successfully drives the robot through the two obstacles, as shown in Fig.~\ref{fig:ellipsoid_vs_scalar}.

\subsubsection{Computational Efficiency}
\begin{figure}[thb]
	\centering
	\includegraphics[width=.7\columnwidth]{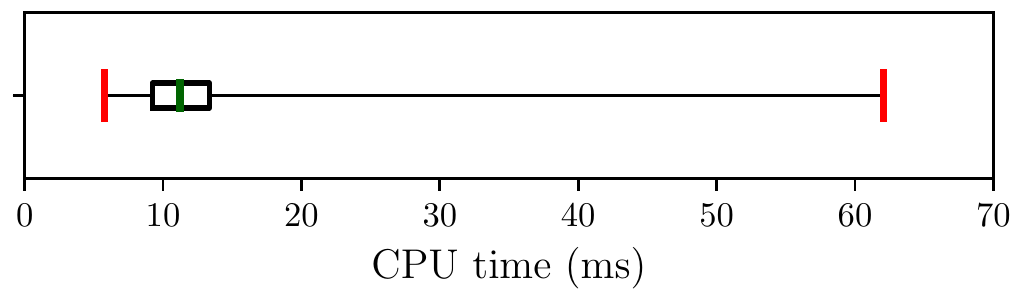}
	\caption{CPU time of solving one zoRO MPC problem.
	\yfHighlight{The visualized are the minimal, the lower quartile, the median, the upper quartile and the maximum.}}
	\label{fig:rmpc_elapsed_time}
\end{figure}
Fig.~\ref{fig:rmpc_elapsed_time} reports the CPU time of solving one robust MPC problem.
The median and the maximum of the running time are 12 milliseconds and 62 milliseconds respectively.